\documentclass[sigconf,nonacm]{acmart}

\setcopyright{none}

\usepackage{balance} 

\usepackage{enumitem}

\usepackage{tikz}
\usetikzlibrary{calc}
\usetikzlibrary{patterns}
\tikzstyle{mybox} = [draw=black, fill=white,  thick,
    rectangle, inner sep=10pt, inner ysep=20pt]
\tikzstyle{mybox} = [draw=black, fill=white,  thick,
    rectangle, inner sep=2pt, inner ysep=2pt]

\usepackage{algorithmic}
\usepackage{algorithm}

\title{Robust Multiagent Collaboration Through Weighted Max–Min 
T-Joins}

\author{Sharareh Alipour}
\email{sharareh.alipour@gmail.com}

\begin{abstract}
Many multiagent tasks---such as reviewer assignment, coalition formation, or fair resource allocation---require selecting a group of agents such that collaboration remains effective even in the worst case. The \emph{weighted max--min $T$-join problem} formalizes this challenge by seeking a subset of vertices whose minimum-weight matching is maximized, thereby ensuring robust outcomes against unfavorable pairings.

We advance the study of this problem in several directions. First, we design an algorithm that computes an upper bound for the \emph{weighted max--min $2k$-matching problem}, where the chosen set must contain exactly $2k$ vertices. Building on this bound, we develop a general algorithm with a \emph{$2 \ln n$-approximation guarantee} that runs in $O(n^4)$ time. Second, using ear decompositions, we propose another upper bound for the weighted max--min $T$-join cost. We also show that the problem can be solved exactly when edge weights belong to $\{1,2\}$.

Finally, we evaluate our methods on real collaboration datasets. Experiments show that the lower bounds from our approximation algorithm and the upper bounds from the ear decomposition method are consistently close, yielding empirically small constant-factor approximations. Overall, our results highlight both the theoretical significance and practical value of weighted max--min $T$-joins as a framework for fair and robust group formation in multiagent systems.

\end{abstract}

\keywords{$T$-join, Weighted max-min $T$-join, Weighted max-min $2k$-matching, Approximation algorithms}

\newcommand{\BibTeX}{\rm B\kern-.05em{\sc i\kern-.025em b}\kern-.08em\TeX}

\begin{document}


\pagestyle{fancy}
\fancyhead{}


\maketitle

\section{Introduction}

The $T$-join problem is an important generalization of matching problems. 
Let $T$ be an even-sized subset of a finite metric space $X$.  
A matching of $T$ pairs up its elements, and the cost of the matching is the sum of the distances within each pair.  
The classical $T$-join problem seeks the matching of $T$ with minimum total cost.  
The max–min $T$-join problem, in contrast, looks for an even-sized subset $T$ of $X$ whose minimum-cost matching is as large as possible.  
We denote this maximum of minimum costs by $\mu(X)$.

Let $G$ be a connected weighted graph with positive edge weights.  
This graph induces a metric on the vertex set $V$, where the distance between two vertices is the weight of the shortest path connecting them.

In such a graph, a subgraph $J$ of $G$ is called \emph{valid} if every cycle $C$ in $G$ satisfies  
\[
w(C) \ge 2\, w(C \cap J),
\]
where $w(C)$ is the total weight of the edges in $C$, and $w(C \cap J)$ is the total weight of the edges of $J$ that lie in $C$.
 
The max–min $T$-join problem for $G$ is equivalent to finding the maximum-weight valid subgraph of $G$.  
The equivalence works as follows: to each valid subgraph $J$, we associate the set of vertices that have odd degree in $J$.  
This set always has even size, and the weight of $J$ equals the cost of the minimum matching of these vertices in $G$.  
Hence, the maximum weight of a valid subgraph corresponds to the maximum–minimum $T$-join of $G$.

Following questions raised by Sol\'e and Zaslawsky~\cite{Sole-Zaslawsky}, Frank~\cite{frank93} studied the case where all edges have weight~$1$ and designed a polynomial-time algorithm running in $O(|V||E|)$ time.  
Iwata and Ravi~\cite{iwata2013tjoin} later extended the problem to general edge weights and developed a constant-factor approximation algorithm.

Another extension of the weighted max--min $T$-join problem is the \emph{weighted max--min $2k$-matching problem}, where the goal is to select $2k$ vertices of the graph so that their minimum weight matching is maximized.  
This problem naturally arises in multiagent systems.  
For example, a coordinator may need to preselect $2k$ responders before pairing them into buddy teams, ensuring that even under worst-case pairings the overall effectiveness remains high.  
In social or collaboration networks, one may wish to form groups where even the weakest partnership has sufficient strength, guaranteeing cohesion.  
In transportation and infrastructure planning, the objective identifies fragile subsets of a network whose worst-case connectivity cost is maximized, guiding the design of resilient routes.  
Finally, in fair task or resource allocation, maximizing the minimum weight matching ensures that no agent receives an extremely weak or undesirable assignment, aligning with fairness principles.

\subsection{Related works}

Our work lies at the intersection of robustness, matching, and subset selection in multiagent systems.    
Selecting a subset of agents that performs well under worst-case conditions relates to robust optimization, coalition formation, and fairness in multiagent systems.  

Okimoto et al.\ introduced Task-Oriented Robust Team Formation, defining \(k\)-robust teams that remain functional after any \(k\) agents drop out~\cite{okimoto15robustteam}, while Schwind et al.\ developed Partial Robustness in Team Formation (PR-TF), which allows partial recovery rather than full robustness~\cite{schwind21partial}.  
In matching under uncertainty, Robust Popular Matchings require a matching to remain popular under perturbations of preferences~\cite{bullinger24robustpopular}.  
Online matching under distributional drift considers robustness to changes in arrival distributions by maintaining worst-case guarantees over time~\cite{zhou2019robustonline}.  
In robust subset and feature selection, Robust Subset Selection by Greedy and Evolutionary Pareto Optimization maximizes the minimum performance across multiple objectives under uncertainty~\cite{bian2022robustsubset}.  
On the repairability side, Recoverable Team Formation studies teams that can be repaired after agent failures at minimal cost~\cite{demirovic18recoverable}.  
Work on team formation also explores compatibility and structure; for example, Towards Realistic Team Formation in Social Networks Based on Densest Subgraphs models compatibility via social links and formulates the team selection as a densest subgraph problem~\cite{rangapuram15team}.  
In peer review and assignment, recent work such as PeerReview4All provides fair and accurate reviewer assignment under load and coverage constraints~\cite{StelmakhSS19}.  
Taken together, these works show the growing interest in combining resilience, fairness, and assignment in AI systems; however, unlike them, our focus is on choosing an optimal subset so that the minimum matching weight within that subset is maximized under worst-case conditions.

\subsection{Our approach and new results}
This paper studies the weighted max--min $T$-join problem.  
We first present a simple greedy algorithm that provides an upper bound for the cost of the weighted max--min $2k$-matching problem in a metric space, 
inspired by the greedy $k$-center clustering algorithm of \cite{gonzalez85kcenter}.  
Given a set of points $P = \{p_1, \dots, p_n\}$ in a metric space, the algorithm starts with an arbitrary point 
and, in each iteration, adds the point farthest from the current set.  
This process defines an ordering of the vertices, and we show how the selected points relate to the optimal solution 
of the weighted max--min $2k$-matching problem.  
As a result, we obtain the first nontrivial upper bound for the weighted max--min $2k$-matching problem.  
This provides initial progress on an open question posed by \cite{iwata2013tjoin}, who gave a constant-factor approximation algorithm for the weighted max--min $T$-join problem but left the $2k$-matching case open.  
Furthermore, our approach also yields a logarithmic-approximation algorithm for the weighted max--min $T$-join problem 
with faster running time.  
The algorithm itself is simple, but the analysis of its approximation factor is more involved.  

Next, we study the weighted max--min $T$-join problem using a different approach.  
Our method is based on \emph{ear decompositions} of the graph and selects edges according to this decomposition 
and specific rules.  
This gives an upper bound, extending the method of \cite{frank93} from the unweighted to the weighted case.  

We also provide an exact algorithm for the weighted max--min $T$-join problem on a special class of graphs 
where edge weights are restricted to $1$ or $2$.  
This class is closely related to the $(1,2)$-TSP, a well-studied variant of the Traveling Salesperson Problem.  

Finally, we test our algorithms on real datasets.  
By combining the lower bound from our approximation algorithm with the upper bound from the ear decomposition method, the results show that the ratio between the two bounds is small in practice.

Both of our bounds can be computed in $O(n^4)$ time, whereas the algorithm of \cite{iwata2013tjoin} relies on the ellipsoid method.  
Consequently, our approach offers a more explicit and practically efficient running time.  
This demonstrates that, in practice, our algorithms are not only faster but also yield more effective results.  

\paragraph{Notations.}
In this paper, we assume all graphs are weighted.  
We denote by $\mu(G)$ the cost of the weighted max--min $T$-join in a graph $G$.  
For the weighted max--min $2k$-matching problem, the cost is written as $\mu_{2k}(G)$.  
Clearly,
\[
\mu(G) = \max_{1 \leq i \leq \lfloor n/2 \rfloor} \mu_{2i}(G).
\]

The $n$th harmonic number is
\[
H_n = 1 + \frac{1}{2} + \frac{1}{3} + \dots + \frac{1}{n}.
\]

\section{From the weighted max--min $2k$-matching problem to the weighted max-min $T$-join}
\label{secMMWM}

We start by presenting an upper bound for the cost of weighted max--min $2k$-matching problem in general metric spaces; 
that is, we assume the edge weights of $G= (V,E)$ satisfy the triangle inequality and that the distance between 
any two vertices $u$ and $v$, $d(u,v)$, is the length of the shortest path between them. 
Given $k$, the algorithm works as follows. Start with any point $v_1 \in V$ and set $C = \{v_1\}$. 
At each step, add to $C$ the point farthest from $C$. 
The distance of a point $p$ from $C$, $d(v,C)$, is defined as
\begin{eqnarray*}
  d(v, C) = \min_{v_i \in C} d(v, v_i).
\end{eqnarray*}

Continue this process until all vertices are added, yielding an ordering of the vertices denoted by $v_1, v_2, \dots, v_n$.

For a set of vertices $A\subseteq V(G)$ with an even number of elements, let $mwm(A)$ denote the cost of the minimum-weight perfect matching on $A$. 
We define
\[
    opt_{2k} = \max_{1 \leq i \leq k} mwm(v_1, \dots, v_{2i}).
\]
Obviously $opt_{2}\leq opt_{4}\leq \dots \leq opt_{2k}$.

We obtain the following upper bound for the cost of $\mu_{2k}(G)$.

\begin{theorem}
\label{mmwm}

Suppose $\{v'_1, \dots, v'_{2k}\} \subseteq V(G)$ is an optimal set for the weighted max--min $2k$-matching problem, 
i.e., $\mu_{2k}(G) = mwm(v'_1, \dots, v'_{2k})$  

Then,
\[
     \mu_{2k}(G)=mwm(v'_1, \dots, v'_{2k}) \;\leq\; 
    2(1+H_{k-1}) opt_{2k}.
\]

\end{theorem}
\begin{proof}

The proof is by induction on $k$. For $k=1$, the optimal solution to the weighted max--min $2k$-matching problem 
is the diameter of $\{v_1, \dots, v_n\}$; that is, the distance between the farthest pair of vertices.
  
Suppose $v'_i$ and $v'_j$ are the endpoints of the diameter of $\{v_1, \dots, v_n\}$, 
i.e., the optimal solution for the max--min $2$-matching problem. Then
\[
d(v_1, v_2) \;\leq\; d(v'_i, v'_j) \;\leq\; d(v'_i, v_1) + d(v_1, v'_j) \;\leq\; 2d(v_1, v_2).
\]

Now, let $\{v'_1, \dots, v'_{2k}\}$ denote the optimal solution for the weighted max--min $2k$-matching problem, i.e., $\mu_{2k}(G) = mwm(v'_1, \dots, v'_{2k})$. By the induction hypothesis,
\[
mwm(v'_1, \dots, v'_{2k}) \;\leq\; 
2(1 + H_{k-1})\, opt_{2k}.
\]

We now prove the statement for $k+1$. 
Assume the optimal solution for the weighted max--min $(2k+2)$-matching problem is 
$\{v''_1, \dots, v''_{2k+2}\}$, i.e, $\mu_{2k+2}(G)=mwm(v''_1, \dots, v''_{2k+2})$.
Without loss of generality, assume that $d(v''_{2k+1}, v''_{2k+2})$ is the smallest distance among all pairs 
$\{v''_i, v''_j\} \subseteq \{v''_1, \dots, v''_{2k+2}\}$. 
Then we clearly have
\[
mwm(v''_1, \dots, v''_{2k+2}) \;\leq\; mwm(v''_1, \dots, v''_{2k}) + d(v''_{2k+1}, v''_{2k+2}).
\]

Also, since the optimal solution for the weighted max--min $2k$-matching problem is 
$\{v'_1, \dots, v'_{2k}\}$, we have 
\[
mwm(v''_1, \dots, v''_{2k}) \;\leq\; mwm(v'_1, \dots, v'_{2k}).
\]
By the induction hypothesis,
\[
mwm(v'_1, \dots, v'_{2k}) \;\leq\; 
2 (1+H_{k-1}) opt_{2k}.
\]

Moreover, by definition of $opt_{2i}$, we have $opt_{2k} \leq opt_{2k+2}$. 
Therefore,
\begin{eqnarray}
\label{aa}
mwm(v''_1, \dots, v''_{2k+2}) \;\leq\; 
2(1+H_{k-1}) opt_{2k+2} 
+ d(v''_{2k+1}, v''_{2k+2}).
\end{eqnarray}

Now, we provide an upper bound for $d(v''_{2k+1}, v''_{2k+2})$. 
According to our greedy algorithm, the order of the selected vertices is $v_1, \dots, v_{2k}$.
Suppose we cluster the points $v''_1, \dots, v''_{2k+2}$ so that each $v''_i$ is assigned to its 
closest point in $\{v_1, \dots, v_{2k}\}$. 
Then there exist two points, $v''_i$ and $v''_j$, that are both assigned to the same vertex, say $v_t \in \{v_1, \dots, v_{2k}\}$.
Hence,
\[
d(v''_{2k+1}, v''_{2k+2}) \;\leq\; d(v''_i, v''_j) 
\;\leq\; d(v''_i, v_t) + d(v_t, v''_j).
\]

According to our greedy algorithm, 
the distance from each point $u' \in \{v_{2k+1}, \dots, v_n\}$ to its closest point in $\{v_1, \dots, v_{2k}\}$ is smaller than the distance between any pair of points in $\{v_1, \dots, v_{2k}\}$. 
Otherwise, $u'$ would have been included among the points in $\{v_1, \dots, v_{2k}\}$. So, we have
\[
d(v''_i, v_t) \;\leq\; d(v_u, v_v) 
\quad \text{and} \quad 
d(v''_j, v_t) \;\leq\; d(v_u, v_v)
\]
for any $v_u, v_v \in \{v_1, \dots, v_{2k}\}$. Thus,

\[
d(v''_i, v_t) + d(v_t, v''_j) \;\leq\; 2d(v_u, v_v).
\]

Now let $d(v_u, v_v)$ be the weight of the smallest edge in the minimum weight perfect matching 
of $\{v_1, \dots, v_{2k}\}$. Then
\begin{eqnarray}
\label{bb}
d(v_u, v_v) \;\leq\; \frac{mwm(v_1, \dots, v_{2k})}{k} 
\;\leq\; \frac{opt_{2k}}{k} 
\;\leq\; \frac{opt_{2k+2}}{k}.
\end{eqnarray}

Therefore, by \ref{aa} and \ref{bb}, we have
\[
mwm(v''_1, \dots, v''_{2k+2}) \;\leq\ 2(1+H_{k}) opt_{2k+2}.
\]

which completes the proof.
\end{proof}
Theorem~\ref{mmwm} provides an upper bound for the cost of weighted max--min $2k$-matching problem. 
Starting from $v_1$, in each step we select the farthest point. Repeating this for $2k$ steps takes $O(nk)$ time. 

Computing the minimum weight perfect matching for a set 
$\{v_1, \dots, v_{2i}\}$ takes $O(i^3)$ time. So, overall, the running time for computing this upper bound is $O(nk + k^4)$. 

Note that the algorithm produces a set of size at most $2k$ whose minimum-weight matching gives an upper bound to $\mu_{2k}(G)$. 

Using Theorem~\ref{mmwm}, we present a logarithmic-factor approximation algorithm 
for the weighted max--min $T$-join problem.

\begin{theorem}
\label{main11}
Let 
\[
mwm(v_1, \dots, v_{2t}) = \max_{1 \leq i \leq \lfloor n/2 \rfloor} mwm(v_1,v_2, \dots, v_{2i}),
\]
and suppose $\{v^*_1, \dots, v^*_{2\ell}\}$ is the optimal solution for the weighted max--min $T$-join problem. 
Then
\begin{align*}
mwm(v_1, \dots, v_{2t})&\leq mwm(v^*_1, \dots, v^*_{2\ell}) \\
&\leq 2 (1+H_{\lfloor n/2 \rfloor-1}) mwm(v_1, \dots, v_{2t}).
\end{align*}

    \end{theorem}

\begin{proof}
Since the minimum weight matching of any set of points is less than $\mu(G)=mwm(v_1^*,\dots,v_{2\ell})$, so clearly we have
\[
mwm(v_1,\dots,v_{2t})\leq mwm(v_1^*,\dots,v^*_{2\ell}).
\]

Now we prove the other inequality.

  Using Theorem~\ref{mmwm}, we have
\[
mwm(v^*_1, \dots, v^*_{2\ell}) \;\leq\; 
2(1+H_{\ell-1}) opt_{2\ell}.
\]

Of course, the value of $\ell$ is unknown. 
Therefore, if we choose
\[
mwm(v_1, \dots, v_{2t}) = \max_{1 \leq i \leq \lfloor n/2 \rfloor} mwm(v_1, \dots, v_{2i}),
\]
then
\begin{align*}
mwm(v_1, \dots, v_{2\ell}) 
&\;\leq\; mwm(v_1, \dots, v_{2t}) 
   \;\leq\; mwm(v^*_1, \dots, v^*_{2\ell}) \\[6pt]
&\;\leq\; 2(1+H_{\ell-1}) opt_{2\ell} \\[6pt]
&\;\leq\; 2(1+H_{\ell-1}) mwm(v_1, \dots, v_{2t}) \\[6pt]
&\leq\; 2(1+H_{\lfloor n/2 \rfloor-1} )opt_{2\lfloor n/2 \rfloor}.
\end{align*}
So, we have 
\[
 mwm(v^*_1, \dots, v^*_{2\ell})\leq 2(1+H_{\lfloor n/2 \rfloor-1}) mwm(v_1, \dots, v_{2t}).
\]
This completes the proof.
\end{proof}

According to Theorem~\ref{main11}, we need to compute $opt_{2\lfloor n/2 \rfloor}$, 
so our algorithm runs in $O(n^4)$ time.  
As mentioned in the introduction, the weighted max--min $T$-join problem also admits 
a constant-factor approximation algorithm~\cite{iwata2013tjoin}, 
but it relies on the ellipsoid method and its running time is not stated explicitly. 
Our algorithm is much simpler, and experiments show that the values of $opt_{2k}$ 
provide tighter bounds in practice.

\section{Ear decomposition and weighted max–min $T$-join problem}
\label{secupperbound}

In this section, we present an algorithm that computes an upper bound for the weighted max--min $T$-join problem.  
This bound applies even when the points do not form a metric space.  

Given a connected weighted graph $G$, the goal is to choose a set of edges with maximum total weight, subject to the condition that in every cycle of the graph, the total weight of the chosen edges is at most half of the total weight of that cycle.  
 According to \cite{iwata2013tjoin}, this is equivalent to selecting a set of vertices whose minimum weight matching is maximized.  

We call any set of edges satisfying the above condition a \emph{valid set of edges}.  
Thus, the problem is to find a valid set of edges with maximum total weight.  
By this definition, any bridge edge (that is, an edge whose removal disconnects the graph) must be included in every optimal valid set of edges.  

Our approach extends the work of Frank~\cite{frank93}, who studied the unweighted version of this problem using ear decompositions of graphs.

We recall the necessary definitions.  
An \emph{ear decomposition} of $G$ is a sequence of subgraphs 
$G_0, G_1, \dots, G_t = G$, where $G_0$ consists of a single vertex and no edges, 
and each $G_i$ is obtained from $G_{i-1}$ by adding a path or a cycle $P_i$.  
The endpoints of $P_i$ (which coincide when $P_i$ is a cycle) lie in $G_{i-1}$, while its remaining vertices are new.  
Each $P_i$ is called an \emph{ear}; if $P_i$ is a single edge, it is called \emph{trivial}.  
The collection $P = \{P_1, \dots, P_t\}$ is also referred to as an ear decomposition, 
and the length of an ear is the number of its edges.  
It is well known that a graph admits an ear decomposition if and only if it is 2-edge-connected.  
A graph is 2-edge-connected if it is connected and remains connected after the removal of any single edge—equivalently, it is a connected graph with no bridge edges.

A graph $G$ is called \emph{factor-critical} if, for every vertex $v \in V(G)$, the graph $G - v$ has a perfect matching.  
The following classical result is due to Lovász~\cite{lovasz1972note}.

\begin{theorem}[\cite{lovasz1972note}]
\label{upperbound}
A graph $G$ is factor-critical if and only if it possesses an odd ear decomposition.  
Moreover, for any edge $e$ of a critical graph $G$, there exists an odd ear decomposition 
of $G$ (i.e, each path has an odd number of edges) in which the first ear contains $e$.
\end{theorem}

Using this characterization, Frank \cite{frank93} proved that for an unweighted graph $G$ (i.e. the weight of each edge is $1$), one has  
\[
\mu(G) = \frac{\varphi(G) + |V| - 1}{2},
\]
where $\varphi(G)$ denotes the minimum number of edges whose contraction yields a factor-critical graph.
Frank’s proof begins by showing that 

\begin{eqnarray}
\label{f}
\varphi(G) \geq 2\mu(G) - |V| + 1.  
\end{eqnarray}

We extend this result to the weighted graphs as follows.  
Given a graph $G$, we first contract all bridges to get a 2-edge-connected graph $G''$. 

If we compute the optimal solution for the weighted max--min $T$ problem on $G''$,  
then adding back the bridge edges gives the optimal solution for $G$.  
Thus, without loss of generality, we may assume that $G$ is 2-connected,  
and therefore admits an ear decomposition.

In an ear decomposition of $G$, for an ear $P$, let $w(P)$ be the sum of its edge weights.
We start with the following lemma.
\begin{lemma}
\label{le1}
Let $E' = \{e_1, \dots, e_t\}$ be a valid set of edges of the graph $G$, and let $u$ and $v$ be two vertices of $G$.  
Choose $P$ to be a path from $u$ to $v$ (or a cycle if $u = v$) such that 
\[
(2\sum_{e \in P \cap E'} w(e)) - w(P)
\]
is minimized among all possible paths from $u$ to $v$ (or cycles if $u = v$).  
Let $E_1 = P \cap E'$ and $E_2 = P \setminus E_1$.  
Then $(E' \setminus E_1) \cup E_2$ is also a valid set.

\end{lemma}
\begin{proof}
Assume that $(E'\setminus E_1)\cup E_2$ is not valid. This mean that there is a cycle that crosses the edges of $P$ such that is violates the condition.
 So, we can find another path $P'$ connecting $u$ and $v$ such that $2\sum_{e\in P'\cap E'(G)}-w(P')$ is less than  $2\sum_{e\in P\cap E'(G)}-w(P)$ which is a contradiction,

\end{proof}

Now, for a given path $P$, define $\max(P)$ as the largest sum of edge weights from $P$ that does not exceed $w(P)/2$.  
If $\{P^*_1, \dots, P^*_t\}$ is an ear decomposition of $G$ minimizing $\sum_{i=1}^t \max(P^*_i)$,  
then we have 
\[
\mu(G) \leq \sum_{i=1}^t \max(P^*_i).
\]
We show this inequality by proving the following theorem.

\begin{theorem}
\label{upperbound}
Let $\{G_0, \dots, G_t = G\}$ be an ear decomposition of $G$,  
where each $G_i$ is obtained from $G_{i-1}$ by adding a path $P_i$.  
Then,
\[
\mu(G) \;\leq\; \sum_{i=1}^t \max(P_i).
\]

\end{theorem}

\begin{proof}
The proof is by induction on the number of ears.  
For $t = 1$, we clearly have $\mu(G) = \max(P_1)$.  

Assume the claim holds for all ear decompositions with at most $k$ ears.  
We prove it for $t = k+1$.  
Let $G = (V,E)$ be obtained from $G' = (V',E')$ by adding the ear $P_{k+1}$, i.e., $G' = G_k$.  

By the induction hypothesis,
\[
\mu(G') \;\leq\; \sum_{i=1}^{k} \max(P_i).
\]
Thus, it is enough to show that
\[
\mu(G') \;\geq\; \mu(G) - \max(P_{k+1}).
\]

Let $u$ and $v$ be the (not necessarily distinct) endpoints of $P_{k+1}$.  
Let 
\[
E^*(G) = \{e^*_1, \dots, e^*_{t'}\}
\] 
be an optimal solution to the weighted max--min  $T$-join problem for $G$, i.e., a valid set of edges with maximum total weight.  

We now consider two cases:

\begin{itemize}
    \item If 
    \[
    \sum_{e \in P_{k+1} \cap E^*(G)} w(e) \;\leq\; \frac{w(P_{k+1})}{2},
    \]
    then in this case $\{e\in P_{k+1} \cap E^*(G)\}$ gives a valid set of edges for $P_{k+1}$, so we have
    \[
    \sum_{e \in P_{k+1} \cap E^*(G)} w(e) \;\leq\; \max(P_{k+1}).
    \]

    By taking the edges in $E^*(G) \cap E(G')$, we obtain a valid set of edges for $G'$,  
    with cost at least $\mu_w(G) - \max(P_{k+1})$.  
    Hence $\mu(G') \geq \mu(G) - \max(P_{k+1})$.  

     \item If 
    \[
    \sum_{e \in P_{k+1} \cap E^*(G)} w(e) \;>\; \frac{w(P_{k+1})}{2},
    \]
    then in the optimal solution for $G$, any cycle formed by $P_{k+1}$ and another path $P'$ 
    connecting $u$ and $v$ must satisfy
    \[
    \sum_{e \in P_{k+1} \cap E^*(G)} w(e) \;+\; \sum_{e \in P' \cap E^*(G)} w(e) 
    \;\leq\; \tfrac{1}{2}\big(w(P_{k+1}) + w(P')\big).
    \]
    Now consider a path $P''$ connecting $u$ and $v$ in $G'$, chosen to minimize 
    \[
   ( 2\sum_{e \in P'' \cap E^*(G)} w(e) )- w(P'').
    \]

    Let $E_1=e \in P'' \cap E^*(G)$ and $E_2=E(P'')/E_1$.
    By Lemma \ref{le1}, replacing the edges of $E_1$ with $E_2$ in $E^*(G)$, gives a valid selection of edges in $G'$. The cost of these edges is at least $\mu(G) - \max(P_{k+1})$. So, we have $\mu(G')\geq \mu(G) - \max(P_{k+1})$.
\end{itemize}

Thus, in both cases, $\mu(G') \geq \mu(G) - \max(P_{k+1})$.  
By the induction hypothesis, the theorem follows.
\end{proof}

If, in Theorem~\ref{upperbound}, we set all edge weights to $1$, then we obtain Inequality~\ref{f}. 
This is because, for an even ear $P_{\text{even}}$, we have $\max(P_{\text{even}}) = \frac{|P_{\text{even}}|}{2}$, 
and for an odd ear $P_{\text{odd}}$, we have $\max(P_{\text{odd}}) = \frac{|P_{\text{odd}}| - 1}{2}$.

According to Theorem~\ref{upperbound}, for a given graph $G$, we can compute an upper bound by first contracting the bridge edges, then constructing an ear decomposition, and finally computing $\max(P)$ for each ear $P$.
Note that in general, given a path $P$, computing $\text{max}(P)$ is NP-hard, as it is equivalent to the knapsack problem \cite{garey1979computers,karp1972reducibility}. However, there exists an FPTAS \cite{ibarra1975fast,vazirani2001approximation} with running time $O(n^2/\epsilon)$ that can compute a $(1-\epsilon)$-approximation for $\max(P)$ for any $\epsilon > 0$.

\paragraph{Remark.}
By Theorem~\ref{upperbound}, any ear decomposition of a graph $G$ gives an upper bound for $\mu(G)$.  
The open problem is to find an optimal ear decomposition $\{P^*_1, \dots, P^*_t\}$ that minimizes
\[
\sum_{i=1}^t \max(P_i).
\]

In the unweighted case, Frank \cite{frank93} proved that
\[
\mu(G) = \sum_{i=1}^t \max(P^*_i),
\]
and also gave a polynomial-time algorithm to compute this value.  
But, in the weighted case, equality may not hold. For example, in Figure~\ref{ex}, the graph $G$ has $\mu(G)=9+\epsilon$, while the optimal ear decomposition gives $9+2\epsilon$.  
Thus, unlike the unweighted case, the optimal ear decomposition in the weighted setting provides only an upper bound for $\mu(G)$.

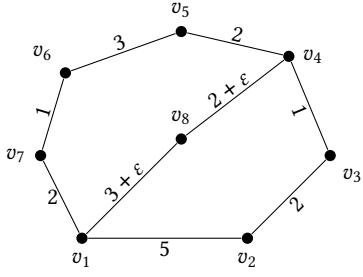
\begin{figure}
\begin{tikzpicture}[scale=1.1,
  vertex/.style={circle,fill,inner sep=1.5pt},
  weight/.style={fill=none,draw=none,inner sep=1pt}]

  \node[vertex,label=below:$v_1$] (v1) at (0,0) {};
  \node[vertex,label=below:$v_2$] (v2) at (2,0) {};
  \node[vertex,label=below right:$v_3$] (v3) at (3,1) {};
  \node[vertex,label=right:$v_4$] (v4) at (2.5,2.2) {};
  \node[vertex,label=above:$v_5$] (v5) at (1.2,2.5) {};
  \node[vertex,label=above left:$v_6$] (v6) at (-0.2,2) {};
  \node[vertex,label=left:$v_7$] (v7) at (-0.5,1) {};
  \node[vertex,label=above:$v_8$] (v8) at (1.2,1.2) {}; 

  \draw (v1) -- node[weight,below] {5} (v2);
  \draw (v2) -- node[weight,sloped,below] {2} (v3);
  \draw (v3) -- node[weight,sloped,below] {1} (v4);
  \draw (v4) -- node[weight,sloped,above] {2} (v5);
  \draw (v5) -- node[weight,sloped,above] {3} (v6);
  \draw (v6) -- node[weight,sloped,above] {1} (v7);
  \draw (v7) -- node[weight,left] {2} (v1);

  \draw (v1) -- node[weight,sloped,above] {$3+\varepsilon$} (v8);
  \draw (v8) -- node[weight,sloped,above] {$2+\varepsilon$} (v4);

\end{tikzpicture}

\caption{An example of a graph $G$ such that $\mu(G)=9+\epsilon\neq \sum \max(P_i^*)=9+2\epsilon$, where $\epsilon$ is a small constant number . }

\label{ex}

\end{figure}

\section{Exact algorithm for weighted max–min $T$-join in $(1,2)$-graphs}
\label{12graph}
In this section, we study the case where $G$ is a $(1,2)$-graph, that is, a complete graph with edge weights $1$ or $2$.  
For this setting, we present a simple algorithm that computes the exact solution to the weighted max--min $T$-join problem.  
We denote the cost of the weighted max--min $T$-join of a $(1,2)$-graph $G$ by $\mu_{(1,2)}(G)$.

First we present the following theorem.

\begin{theorem}
\label{th12}
 Let $G = (V, E)$ be a $(1,2)$-graph with $n$ vertices.  
If $n$ is even, then $\mu_{1,2}(G) = mwm(v_1, \dots, v_n)$.  
If $n$ is odd, there exists a vertex $v_i \in V$ such that 
$\mu_{1,2}(G) = mwm(v_1, \dots, v_n \setminus \{v_i\})$.
\end{theorem}

\begin{proof}
We show that there exists a solution to the weighted max--min $T$-join problem for a $(1,2)$-graph that uses almost all vertices 
(all vertices if $n$ is even, and all but one if $n$ is odd).  

Suppose $v^*_1, \dots, v^*_{2k}$ is an optimal solution for the weighted max--min $T$-join with the maximum possible number of selected vertices 
among all optimal sets. Then we show that  $2k \geq n-1$. Otherwise, we could add two more vertices $v_i$ and $v_j$ to 
$\{v^*_1, \dots, v^*_{2k}\}$ such that
\[
mwm(v^*_1, \dots, v^*_{2k}, v_i, v_j) = mwm(v^*_1, \dots, v^*_{2k}),
\]
which contradicts the maximality of $2k$.

Assume that $2k< n-1$.
Choose two arbitrary vertices not in $\{v^*_1, \dots, v^*_{2k}\}$. 
There are two cases:
\begin{enumerate}
    \item $v_i$ and $v_j$ are matched together in the minimum weight perfect matching of 
$\{v^*_1, \dots, v^*_{2k}, v_i, v_j\}$ in $G$. 
This is impossible, since in that case
\[
mwm(v^*_1, \dots, v^*_{2k}, v_i, v_j) 
= d(v_i, v_j) + mwm(v^*_1, \dots, v^*_{2k}),
\]
which contradicts the optimality of $\{v^*_1, \dots, v^*_{2k}\}$.
\item  $v_i$ and $v_j$ are not matched together in the minimum weight perfect matching of 
$\{v^*_1, \dots, v^*_{2k}, v_i, v_j\}$. 
Without loss of generality, assume $v_i$ is matched with $v^*_1$ and $v_j$ with $v^*_2$. 
Then
\[
mwm(v^*_1, \dots, v^*_{2k}, v_i, v_j) 
= d(v^*_1, v_i) + d(v^*_2, v_j) + mwm(v^*_3, \dots, v^*_{2k}).
\]
Thus, since $1\leq d(v^*_1, v_i)$ and $1\leq d(v^*_2, v_j)$,
\[
2 + mwm(v^*_3, \dots, v^*_{2k}) \;  \leq\; mwm(v^*_1, \dots, v^*_{2k}, v_i, v_j).
\]

Since $\{v^*_1,
\dots, v^*_{2k}\}$ is optimal solution for weighted max-min $T$-join problem, we must also have

\begin{eqnarray*}
2 + mwm(v^*_3, \dots, v^*_{2k})\;  &\leq& mwm(v^*_1,\dots v^*_{2k},v_i,v_i)\\
&\leq&  mwm(v^*_1, \dots, v^*_{2k}).
\end{eqnarray*}

Another possible (non-minimal) matching is to pair $v^*_1$ with $v^*_2$ and then take the 
minimum matching on $\{v^*_3, \dots, v^*_{2k}\}$. This gives

\begin{eqnarray*}
    mwm(v^*_1, \dots, v^*_{2k}) &\leq& d(v^*_1, v^*_2) + mwm(v^*_3, \dots, v^*_{2k})\\
    &\leq& 2 + mwm(v^*_3, \dots, v^*_{2k}).    
\end{eqnarray*}

Hence,
\[
mwm(v^*_1, \dots, v^*_{2k}) = 2 + mwm(v^*_3, \dots, v^*_{2k}).
\]

On the other hand,
\[
mwm(v^*_1, \dots, v^*_{2k}) 
= 2 + mwm(v^*_3, \dots, v^*_{2k}) 
\;\leq\; mwm(v^*_1, \dots, v^*_{2k}, v_i, v_j).
\]

This implies
\[
mwm(v^*_1, \dots, v^*_{2k}) = mwm(v^*_1, \dots, v^*_{2k}, v_i, v_j),
\]
contradicting the assumption that $\{v^*_1, \dots, v^*_{2k}\}$ is an optimal solution of maximum size.  
Therefore, $2k \geq n-1$.
   
\end{enumerate}
 
\end{proof}

So, based on Theorem \ref{th12}, we present the following algorithm for computing $\mu_{1,2}(G)$.

If, $n$ is even
\begin{itemize}
    \item Compute a maximum matching on the vertices with weight $1$. i.e, compute the maximum matching of  the spanning graph with edge $1$ edges.
    \item Compute a maximum matching on the remaining vertices. Obviously the weight of each edge is $2$ in this matching, because otherwise, that edge should have been selected in the first step. 
\end{itemize}
Let the cost of the maximum matching be $m_1$.  
Then the number of unmatched vertices is $n - 2m_1$.  
The cost of the maximum matching on these remaining vertices is  
$\frac{n - 2m_1}{2} $.
This is because all edges among the $n - 2m_1$ vertices must have weight $2$.  
Indeed, if any of these edges had weight $1$, we could add it to the $m_1$ edges already chosen, thereby increasing the cost of the maximum matching on the subgraph of weight-$1$ edges — a contradiction.

Thus, the total cost of the selected edges is  
$n - m_1,$  
which is exactly the cost of the minimum weight matching on all vertices of the graph.

If $n$ is odd, for each vertex $v_i$, remove $v_i$ from the graph and compute the minimum-weight matching for the remaining vertices, as explained for the even case.  
Let $MWM_i$ denote the cost of the minimum-weight matching obtained after removing $v_i$.  
Then,
\[
\mu_{1,2}(G) = \max_{1 \le i \le n} MWM_i.
\]

\paragraph{Remark.}
Although the $(1,2)$-case may appear restricted, it captures a range of practical scenarios. 
$(1,2)$-graphs model situations where relationships are essentially binary but may carry an occasional ``boost.'' 
A weight of $1$ represents a basic connection, while a weight of $2$ reflects a stronger or reinforced interaction. 
Such graphs arise naturally when ties can be classified as weak versus strong, adequate versus highly suitable, or low- versus high-capacity. 
Thus, $(1,2)$-graphs provide a simple yet useful abstraction for applications where interactions come in two discrete levels of strength.

We can also extend the proof of Theorem~\ref{th12} to the case where $G$ is a complete graph whose maximum-to-minimum edge-weight ratio is less than $2$. 
Thus, the result applies to a wide range of practical scenarios.

 \section{Experimental Evaluation}

We evaluate our algorithms on two publicly available datasets. 
In these datasets, vertices represent agents or participants, and edge weights 
capture compatibility, collaboration strength, or assignment utility. 
These instances provide natural scenarios for studying fairness, where the 
objective is to maximize the guaranteed (minimum) quality of pairings.  

For each scenario, we describe the dataset used, how the corresponding graphs 
are generated, and the algorithms applied to compute  approximate 
solutions. We then present and analyze the results, showing how the algorithms 
perform across different settings.

\subsection{Reviewer assignment}
\label{sube}
Large AI conferences often rely on reviewer--paper similarity scores to guide 
the assignment process. Motivated by this setting, we study how to design a 
reviewer assignment procedure. Suppose we have a pool of candidates and wish to 
select a subset as the reviewers for a conference. Our goal is to ensure that 
when each paper is assigned to two reviewers, those reviewers are not too similar. 
This guarantees that at least one reviewer brings a different perspective to the 
paper, improving fairness and reducing bias in the review process.  

In other words, we want to select a set of reviewers such that, even under the 
worst-case matching (where the most similar reviewers are paired), the overall 
matching still has maximum possible size. This ensures that the assignment step 
can maximize the worst-case assignment quality.  

\paragraph{Graph construction.} 
We extracted the DBLP collaboration graph\footnote{\url{https://dblp.org}}, where
vertices are authors and edge weights are their numbers of joint publications.

Because we lack actual reviewer lists, we built a synthetic pool of candidates:
we selected $70$ authors with many publications in top-tier theoretical computer
science conferences, together with their coauthors. Such authors are typically
among program committee members and reviewers.

For each \(n \in \{30,40,50,60,70\}\), we sampled \(n\) authors uniformly at random from this pool and formed the induced subgraph \(G\) on these vertices, but with a modification of the edge weights.
For each pair $v_i$ and $v_j$, 
let $c(v_i,v_j)$ denote the number of papers they have coauthored. We set $d(v_i, v_j) \;=\; \frac{1}{c(v_i,v_j)+1}$.
Smaller weights indicate stronger similarity, and this weighting makes the graph connected. To enforce the triangle inequality, set \(w(v_i,v_j)\) to be the shortest-path distance between \(v_i\) and \(v_j\). Using these distances, we compute upper and lower bounds for \(\mu(G)\).

\paragraph{Lower bound (Theorem~\ref{main11}).}
We implement the procedure from Section~\ref{secMMWM} as follows:
\begin{enumerate}
  \item Run the greedy algorithm from Section~\ref{secMMWM} to obtain an ordering of the vertices.  
  Let 
  \[
  mwm(v_1, \dots, v_{2t}) = \max_{1 \le i \le \lfloor n/2 \rfloor} mwm(v_1, v_2, \dots, v_{2i}),
  \]
  and set $\{v_1, \dots, v_{2t}\}$ as the group of candidate reviewers.  

  \item Record $opt_{2\lfloor n/2 \rfloor}=mwm(v_1, v_2, \dots, v_{2t})$ as the lower bound on $\mu(G)$ guaranteed by Theorem~\ref{main11}.
\end{enumerate}

\paragraph{Upper bound via ear decomposition (Theorem~\ref{upperbound}).}
Since the graph is a weighted complete graph, we construct an explicit ear decomposition as follows:

\begin{enumerate}
  \item Initialize $G_0$ with a single vertex $v_1$.
  \item Let $G_1$ be a Hamiltonian cycle, i.e., a cycle that visits all vertices.
  \item While there is an unused edge $e$ in the graph, let $P_i = e$ and set $G_i = G_{i-1} + P_i$.
\end{enumerate}

This process yields an ear decomposition $\mathcal{E}$ of $G$.  
For each trivial ear $P \in \mathcal{E}$, we have $\max(P) = 0$.  
For $G_1 = G_0 + P_1$, where $P_1$ is the Hamiltonian cycle, $\max(P_1) \leq w(P_1)/2$.  
Thus, $w(P_1)/2$ provides an upper bound on $\mu(G)$.

To obtain a tighter bound, we need a Hamiltonian cycle with the minimum total edge weight.  
This corresponds to the \emph{traveling salesman problem (TSP)}, which is NP-complete~\cite{DBLP:books/fm/GareyJ79} and has been widely studied.  
We use the 1.5-approximation algorithm by Christofides~\cite{christofides1976worst} in our experiments to compute such a cycle.  
Specifically, for each graph $G$, we compute a Hamiltonian cycle using Christofides' algorithm.  
Let $TSP_{1.5}(G)$ denote the cost of this cycle; then $TSP_{1.5}(G)/2$ serves as an upper bound for $\mu(G)$.

The results are presented in Table~\ref{t2}.  
For all datasets, the ratio between the upper and lower bounds is less than $1.19$,  
indicating that, in practice, the set of vertices selected by our greedy algorithm, $\{v_1,\dots,v_{2t}\}$, yields a minimum weight matching whose cost is close to the optimal value.

\begin{table}[h!]
\centering
\begin{tabular}{|c|c|c|c|}
\hline

$n$ & LB $opt_{2\lfloor n/2 \rfloor}$ & UB $\frac{TSP_{1.5}}{2}$& UP/LB \\
\hline
30 & 2.08 & 2.41& 1.16\\
40 & 4.39 & 5.18& 1.18 \\
50 & 6.04 & 6.86& 1.13\\
60 & 7.00 & 8.23&1.17 \\
70 & 7.94 & 9.48& 1.19\\
\hline
\end{tabular}
\caption{Lower and upper bounds on $\mu(G)$ and their ratio for reviewer assignment graphs with varying vertex counts.}
\label{t2}
\end{table}

In the previous experiment, we computed lower and upper bounds for $\mu(G)$.
In the next experiment, we use these bounds to obtain good approximate solutions for $\mu_{2k}(G)$.
So, we are given a set of candidate reviewers and we must select exactly $2k$ of them.  
The goal is to choose these $2k$ reviewers so that the cost of the worst-case matching is maximized.

\begin{table}[h!]
\centering
\footnotesize{
\begin{tabular}{|c|c|c|c|c|c|}
\hline
$2k$ & $mwm(v_1,\dots,v_{2k})$ & $opt_{2k}$&$2(1+H_{k-1})opt_{2k}$&$\frac{TSP_{1.5}}{2}$&UP/LB \\
\hline
 2  & 1.00 & 1.00 & 2.00  & 9.48 & 2.00 \\
 4  & 1.61 & 1.61 & 6.44  & 9.48 & 4.00 \\
 6  & 2.27 & 2.27 & 11.35 & 9.48 & 4.18 \\
 8  & 2.94 & 2.94 & 16.66 & 9.48 & 3.23 \\
10  & 3.39 & 3.39 & 20.90 & 9.48 & 2.80 \\
12  & 3.72 & 3.72 & 24.43 & 9.48 & 2.55 \\
14  & 4.25 & 4.25 & 29.33 & 9.48 & 2.23 \\
16  & 4.75 & 4.75 & 34.13 & 9.48 & 2.00 \\
18  & 5.01 & 5.01 & 37.25 & 9.48 & 1.89 \\
20  & 5.48 & 5.48 & 41.97 & 9.48 & 1.73 \\
22  & 5.76 & 5.76 & 45.26 & 9.48 & 1.65 \\
24  & 5.97 & 5.97 & 48.00 & 9.48 & 1.59 \\
26  & 6.20 & 6.20 & 50.88 & 9.48 & 1.53 \\
28  & 6.57 & 6.57 & 54.93 & 9.48 & 1.44 \\
30  & 6.73 & 6.73 & 57.23 & 9.48 & 1.41 \\
32  & 7.07 & 7.07 & 61.06 & 9.48 & 1.34 \\
34  & 7.19 & 7.19 & 62.77 & 9.48 & 1.32 \\
36  & 7.44 & 7.44 & 65.74 & 9.48 & 1.27 \\
38  & 7.46 & 7.46 & 66.66 & 9.48 & 1.27 \\
40  & 7.50 & 7.50 & 68.22 & 9.48 & 1.26 \\
42  & 7.63 & 7.63 & 70.16 & 9.48 & 1.24 \\
44  & 7.64 & 7.64 & 70.98 & 9.48 & 1.24 \\
46  & 7.64 & 7.64 & 71.68 & 9.48 & 1.24 \\
48  & 7.64 & 7.64 & 72.34 & 9.48 & 1.24 \\
50  & 7.73 & 7.73 & 73.84 & 9.48 & 1.23 \\
52  & 7.74 & 7.74 & 74.55 & 9.48 & 1.22 \\
54  & 7.80 & 7.80 & 75.73 & 9.48 & 1.22 \\
56  & 7.80 & 7.80 & 76.31 & 9.48 & 1.22 \\
58  & 7.90 & 7.90 & 77.85 & 9.48 & 1.20 \\
60  & 7.95 & 7.95 & 78.89 & 9.48 & 1.19 \\
62  & 7.95 & 7.95 & 79.42 & 9.48 & 1.19 \\
64  & 7.92 & 7.95 & 79.93 & 9.48 & 1.20 \\
66  & 7.86 & 7.95 & 80.43 & 9.48 & 1.21 \\
68  & 7.93 & 7.95 & 80.91 & 9.48 & 1.20 \\
70  & 7.92 & 7.95 & 81.38 & 9.48 & 1.20 \\
\hline
\end{tabular}
}
\caption{Lower and upper bounds of $\mu_{2k}(G)$ for the reviewer assignment graph with 70 vertices, computed for different values of $2k$. 
The lower bound is obtained from the minimum-weight matching on $\{v_1, \dots, v_{2k}\}$, 
and the upper bound is $\min\{2(1 + H_{k-1})\,opt_{2k},\, TSP_{1.5}/2\}$. 
The last column shows the ratio of the upper bound to the lower bound.
}

\label{t2k}
\end{table}

In our experiments, we consider a graph $G$ with $70$ vertices corresponding to $70$ authors, constructed as described earlier.  
Using the greedy algorithm (Section~\ref{mmwm}), we obtained an ordering of the vertices of $G$.  
For each $2k = 2, \dots, 70$, we computed the minimum-weight matching on $\{v_1, \dots, v_{2k}\}$.  
Clearly, $mwm(v_1, \dots, v_{2k})$ provides a lower bound for $\mu_{2k}(G)$.  

In Section~\ref{mmwm}, we established the following upper bound:
\[
\mu_{2k}(G) \leq 2(1 + H_{k-1})\, opt_{2k}.
\]
For each $2k = 2, \dots, 70$, we computed $2(1 + H_{k-1})\, opt_{2k}$ as an upper bound.  

Another upper bound for $\mu_{2k}(G)$ can be derived from the cost of a Hamiltonian cycle.  
Let $v^*_1, \dots, v^*_{2k}$ be an optimal solution to the weighted max--min $2k$-matching problem.  
Any Hamiltonian cycle on these vertices can be divided into two disjoint matchings, so the cost of the minimum-weight matching on $\{v^*_1, \dots, v^*_{2k}\}$ is at most half the cost of the Hamiltonian cycle.  

Hence, we seek a Hamiltonian cycle with the minimum total cost, i.e., the solution to the traveling salesman problem (TSP).  
As discussed earlier, we computed a $1.5$-approximation of the TSP on all $70$ points, with a total cost of $18.96$.  
Therefore, $9.49$ serves as an upper bound for $\mu_{2k}(G)$ for all values of $2k$.

Table~\ref{t2k} reports the computed lower and upper bounds for this graph.  
For $2k > 4$, the TSP-based upper bound is tighter than $2(1 + H_{k-1})\, opt_{2k}$.  
As shown, for the weighted max--min $2k$-matching problem, the vertices $v_1, \dots, v_{2k}$ selected by the greedy algorithm yield near-optimal solutions.  
In all cases, the ratio between the smallest upper bound and $mwm(v_1, \dots, v_{2k})$ is less than $4.18$.  
Hence, in practice, the set $\{v_1, \dots, v_{2k}\}$ provides a high-quality solution for the weighted max--min $2k$-matching problem on this graph.

Note that in this dataset, the upper bound given by half of the approximate TSP cost works well. 
However, it can be poor in some cases. 
For example, if \(G\) is a complete graph with \(n\) vertices and all edge weights equal to \(1\), 
then the TSP cost is \(n\), while \(\mu_{2k}(G) = k\). 
Thus, this upper bound is weak, whereas the bound \(2(1 + H_{k+1})\,opt_{2k}\) provides a better estimate.
Therefore, in practice, it is reasonable to use both upper bounds and select the minimum between them.

Also, note that $opt_{2k}$ is not a lower bound for $\mu_{2k}(G)$. 
This is because the number of vertices whose minimum weight matching 
equals $opt_{2k}$ can be less than $2k$. As a result, $opt_{2k}$ may 
be larger than $\mu_{2k}(G)$. 
For example, place $2n$ points on a line with $p_{2i}=i$ and 
$p_{2i+1}=i+\epsilon$. In this case, $\mu_{2n}(G)=n\epsilon$, while 
$opt_{2}=opt_{2n}\geq n/2$, since starting from any point and matching it 
to the farthest point gives a cost of at least $n/2$.
Table \ref{t2k} also shows that $opt_{66} = mwm(v_1,\dots,v_{60}) = 7.95$, 
while $mwm(v_1,\dots,v_{66}) = 7.86$. This means we cannot guarantee 
the existence of $66$ vertices whose minimum weight matching is at least 
$opt_{66}=7.95$, since the value $7.95$ already comes from the matching 
of only $60$ vertices.

\subsection{Fair Mutual-Aid Pairing for Mobile Services}
Let $G=(V,E)$ be a graph where vertices represent cities and each edge $(i,j)$ has a travel distance $d(i,j)$.  
We want to form a matching (pairs of cities) to share a mobile resource (e.g., ambulance, repair crew, or relief team).  
The objective is to maximize the minimum distance between paired cities, so that the pairs are spread out and provide diverse services.  

We define the service-quality weight for each pair $(i,j)$ from their distance as follows:

\[
w(i,j) \;=\text{distance between $i$ and $j$}
\]

We aim to select a set of cities such that, even under the worst-case matching, the overall matching has the largest possible value.  
In fairness terms, this corresponds to maximizing the service quality of the worst-off pair.  

In our experiments, we considered $n = 10, 20, 30, \dots, 70$ major cities worldwide and constructed the corresponding graphs.  
For each graph $G$, we computed both the upper and lower bounds for $\mu(G)$, as described in Subsection~\ref{sube}.  
The results are summarized in Table~\ref{t3}.

\begin{table}[h!]
\centering
\begin{tabular}{|c|c|c|c|}
\hline
$n$ & LB $opt_{2\lfloor n/2 \rfloor}$ & UB $\frac{TSP_{1.5}}{2}$&  UP/LB \\
\hline
10  & 18534 & 20864&1.12 \\
20  & 22638 & 25275 & 1.11\\
30  & 27300 & 34966 &1.28 \\
40  & 30827 & 40844 &1.32 \\
50  & 36900 & 52653 &1.41 \\
60  & 37500 & 55665& 1.48\\
70  & 38200 & 55797&1.46  \\

\hline
\end{tabular}
\caption{Lower and upper bounds on $\mu(G)$ and their ratio for city distance graphs with varying vertex counts.}
\label{t3}

\end{table}

In our next experiment, similar to the reviewer assignment scenario, we considered the case where exactly $2k$ vertices must be selected.  
Using the distance graph of $70$ major cities, we computed both the lower and upper bounds for $\mu_{2k}(G)$.  

For the lower bound, we applied the greedy algorithm to compute $mwm(v_1, \dots, v_{2k})$, where vertices were added greedily for each $2k = 2, \dots, 70$.  
For the upper bound, we computed both $2(1 + H_{k-1})\, opt_{2k}$ and half of the TSP cost obtained from the 1.5-approximation algorithm applied to all vertices.  

The results are summarized in Table~\ref{t2b}. For $2k > 2$, the TSP-based upper bound is tighter than $2(1 + H_{k-1})\, opt_{2k}$.  
As shown, for the weighted max--min $2k$-matching problem, the vertices $v_1, \dots, v_{2k}$ selected by the greedy algorithm yield near-optimal solutions.  
For all values of $2k$, the ratio between the smallest upper bound and $mwm(v_1, \dots, v_{2k})$ is less than $2$ except for $2k=4,6$. 
 Hence, in practice, the set $\{v_1, \dots, v_{2k}\}$ provides an effective solution for the weighted max--min $2k$-matching problem on this dataset.

\begin{table}[h!]
\centering
\footnotesize{
\begin{tabular}{|c|c|c|c|c|c|}
\hline
$2k$ & $mwm(v_1,\dots,v_{2k})$ & $opt_{2k}$&$2(1+H_{k-1})opt_{2k}$&$\frac{TSP_{1.5}}{2}$&UP/LB \\
\hline

 2  & 18567 & 18567 & 37134  & 55797 & 2.00 \\
 4  & 18578 & 18578 & 74310  & 55797 & 3.00 \\
 6  & 24499 & 24499 & 122497 & 55797 & 2.28 \\
 8  & 27103 & 27103 & 153582 & 55797 & 2.06 \\
10  & 30595 & 30595 & 188668 & 55797 & 1.82 \\
12  & 29331 & 30595 & 200906 & 55797 & 1.90 \\
14  & 28365 & 30595 & 211104 & 55797 & 1.97 \\
16  & 28017 & 30595 & 219845 & 55797 & 1.99 \\
18  & 30958 & 30958 & 230192 & 55797 & 1.80 \\
20  & 31076 & 31076 & 237979 & 55797 & 1.80 \\
22  & 33241 & 33241 & 261202 & 55797 & 1.68 \\
24  & 32919 & 33241 & 267246 & 55797 & 1.70 \\
26  & 34158 & 34158 & 280315 & 55797 & 1.63 \\
28  & 34651 & 34651 & 289693 & 55797 & 1.61 \\
30  & 31817 & 34651 & 294643 & 55797 & 1.75 \\
32  & 32696 & 34651 & 299263 & 55797 & 1.71 \\
34  & 33270 & 34651 & 303579 & 55797 & 1.68 \\
36  & 33782 & 34651 & 307616 & 55797 & 1.65 \\
38  & 34152 & 34651 & 311399 & 55797 & 1.63 \\
40  & 34617 & 34651 & 314949 & 55797 & 1.61 \\
42  & 35067 & 35067 & 322458 & 55797 & 1.59 \\
44  & 35223 & 35223 & 327249 & 55797 & 1.58 \\
46  & 35259 & 35259 & 330787 & 55797 & 1.58 \\
48  & 37133 & 37133 & 351598 & 55797 & 1.50 \\
50  & 37355 & 37355 & 356811 & 55797 & 1.49 \\
52  & 38075 & 38075 & 366734 & 55797 & 1.47 \\
54  & 38080 & 38080 & 369713 & 55797 & 1.47 \\
56  & 38269 & 38269 & 374384 & 55797 & 1.46 \\
58  & 36549 & 38269 & 377117 & 55797 & 1.53 \\
60  & 36019 & 38269 & 379757 & 55797 & 1.55 \\
62  & 34234 & 38269 & 382308 & 55797 & 1.63 \\
64  & 34230 & 38269 & 384777 & 55797 & 1.63 \\
66  & 35435 & 38269 & 387169 & 55797 & 1.57 \\
68  & 34158 & 38269 & 389488 & 55797 & 1.63 \\
70  & 33914 & 38269 & 391739 & 55797 & 1.65 \\
\hline
\end{tabular}
}
\caption{Lower and upper bounds of $\mu_{2k}(G)$ for the 70-city distance graph, computed for various values of $2k$. 
The lower bound is obtained from the minimum-weight matching on $\{v_1, \dots, v_{2k}\}$, 
and the upper bound is $\min\{2(1 + H_{k-1})\,opt_{2k},\, TSP_{1.5}/2\}$. 
The last column shows the ratio of the upper bound to the lower bound.
}
\label{t2b}
\end{table}

\subsection{Remark} 

To compute an upper bound for $\mu(G)$, we applied a simple greedy ear decomposition. 
Although more sophisticated decompositions could be considered, our results indicate 
that even this straightforward approach yields strong bounds in practice. 
In practice, one may therefore consider multiple ear decompositions and select the one giving 
the smallest bound.

\bibliographystyle{ACM-Reference-Format} 
\bibliography{sample}

@article{christofides1976worst,
  title={Worst-case analysis of a new heuristic for the traveling salesman problem},
  author={CHRISTOFIDES, N},
  journal={Report 388, Graduate School of Industrial Administration, Carnegie Mellon University},
  year={1976}
}

@book{DBLP:books/fm/GareyJ79,
  author       = {M. R. Garey and
                  David S. Johnson},
  title        = {Computers and Intractability: {A} Guide to the Theory of NP-Completeness},
  publisher    = {W. H. Freeman},
  year         = {1979},
  isbn         = {0-7167-1044-7},
  bibsource    = {dblp computer science bibliography, https://dblp.org}
}

@article{gonzalez85kcenter,
  author       = {Teofilo F. Gonzalez},
  title        = {Clustering to Minimize the Maximum Intercluster Distance},
  journal      = {Theor. Comput. Sci.},
  volume       = {38},
  pages        = {293--306},
  year         = {1985},
  url          = {https://doi.org/10.1016/0304-3975(85)90224-5},
  doi          = {10.1016/0304-3975(85)90224-5},
  bibsource    = {dblp computer science bibliography, https://dblp.org}
}

@article{iwata2013tjoin,
  author       = {Satoru Iwata and
                  R. Ravi},
  title        = {Approximating max-min weighted T-joins},
  journal      = {Oper. Res. Lett.},
  volume       = {41},
  number       = {4},
  pages        = {321--324},
  year         = {2013},
  url          = {https://doi.org/10.1016/j.orl.2013.03.004},
  doi          = {10.1016/J.ORL.2013.03.004},
  bibsource    = {dblp computer science bibliography, https://dblp.org}
}

@article{lovasz1972note,
  title={A note on factor-critical graphs},
  author={Lov{\'a}sz, L{\'a}szl{\'o}},
  journal={Studia Sci. Math. Hungar},
  volume={7},
  number={279-280},
  pages={11},
  year={1972}
}

@article{frank93,
  author       = {Andr{\'{a}}s Frank},
  title        = {Conservative weightings and ear-decompositions of graphs},
  journal      = {Comb.},
  volume       = {13},
  number       = {1},
  pages        = {65--81},
  year         = {1993},
  url          = {https://doi.org/10.1007/BF01202790},
  doi          = {10.1007/BF01202790},
  bibsource    = {dblp computer science bibliography, https://dblp.org}
}

@book{garey1979computers,
  author       = {M. R. Garey and
                  David S. Johnson},
  title        = {Computers and Intractability: {A} Guide to the Theory of NP-Completeness},
  publisher    = {W. H. Freeman},
  year         = {1979},
  isbn         = {0-7167-1044-7},
  bibsource    = {dblp computer science bibliography, https://dblp.org}
}

@inproceedings{karp1972reducibility,
  author       = {Richard M. Karp},
  editor       = {Raymond E. Miller and
                  James W. Thatcher},
  title        = {Reducibility Among Combinatorial Problems},
  booktitle    = {Proceedings of a symposium on the Complexity of Computer Computations,
                  held March 20-22, 1972, at the {IBM} Thomas J. Watson Research Center,
                  Yorktown Heights, New York, {USA}},
  series       = {The {IBM} Research Symposia Series},
  pages        = {85--103},
  publisher    = {Plenum Press, New York},
  year         = {1972},
  url          = {https://doi.org/10.1007/978-1-4684-2001-2\_9},
  doi          = {10.1007/978-1-4684-2001-2\_9},
  bibsource    = {dblp computer science bibliography, https://dblp.org}
}

@article{ibarra1975fast,
  author       = {Oscar H. Ibarra and
                  Chul E. Kim},
  title        = {Fast Approximation Algorithms for the Knapsack and Sum of Subset Problems},
  journal      = {J. {ACM}},
  volume       = {22},
  number       = {4},
  pages        = {463--468},
  year         = {1975},
  url          = {https://doi.org/10.1145/321906.321909},
  doi          = {10.1145/321906.321909},
  bibsource    = {dblp computer science bibliography, https://dblp.org}
}

@book{vazirani2001approximation,
  author       = {Vijay V. Vazirani},
  title        = {Approximation algorithms},
  publisher    = {Springer},
  year         = {2001},
  url          = {http://www.springer.com/computer/theoretical+computer+science/book/978-3-540-65367-7},
  isbn         = {978-3-540-65367-7},
  bibsource    = {dblp computer science bibliography, https://dblp.org}
}

@article{Sole-Zaslawsky,
  author       = {Patrick Sol{\'{e}} and
                  Thomas Zaslavsky},
  title        = {The Covering Radius of the Cycle Code of a Graph},
  journal      = {Discret. Appl. Math.},
  volume       = {45},
  number       = {1},
  pages        = {63--70},
  year         = {1993},
  url          = {https://doi.org/10.1016/0166-218X(93)90140-J},
  doi          = {10.1016/0166-218X(93)90140-J},
  bibsource    = {dblp computer science bibliography, https://dblp.org}
}

@inproceedings{okimoto15robustteam,
   author       = {Tenda Okimoto and
                  Nicolas Schwind and
                  Maxime Clement and
                  Tony Ribeiro and
                  Katsumi Inoue and
                  Pierre Marquis},
  editor       = {Gerhard Weiss and
                  Pinar Yolum and
                  Rafael H. Bordini and
                  Edith Elkind},
  title        = {How to Form a Task-Oriented Robust Team},
  booktitle    = {Proceedings of the 2015 International Conference on Autonomous Agents
                  and Multiagent Systems, {AAMAS} 2015, Istanbul, Turkey, May 4-8, 2015},
  pages        = {395--403},
  publisher    = {{ACM}},
  year         = {2015},
  url          = {http://dl.acm.org/citation.cfm?id=2772931},
  bibsource    = {dblp computer science bibliography, https://dblp.org}
}

@inproceedings{schwind21partial,
author       = {Nicolas Schwind and
                  Emir Demirovic and
                  Katsumi Inoue and
                  Jean{-}Marie Lagniez},
  editor       = {Frank Dignum and
                  Alessio Lomuscio and
                  Ulle Endriss and
                  Ann Now{\'{e}}},
  title        = {Partial Robustness in Team Formation: Bridging the Gap between Robustness
                  and Resilience},
  booktitle    = {{AAMAS} '21: 20th International Conference on Autonomous Agents and
                  Multiagent Systems, Virtual Event, United Kingdom, May 3-7, 2021},
  pages        = {1154--1162},
  publisher    = {{ACM}},
  year         = {2021},
  url          = {https://www.ifaamas.org/Proceedings/aamas2021/pdfs/p1154.pdf},
  doi          = {10.5555/3463952.3464086},
  bibsource    = {dblp computer science bibliography, https://dblp.org}
}

@inproceedings{demirovic18recoverable,
  author       = {Emir Demirovic and
                  Nicolas Schwind and
                  Tenda Okimoto and
                  Katsumi Inoue},
  editor       = {Elisabeth Andr{\'{e}} and
                  Sven Koenig and
                  Mehdi Dastani and
                  Gita Sukthankar},
  title        = {Recoverable Team Formation: Building Teams Resilient to Change},
  booktitle    = {Proceedings of the 17th International Conference on Autonomous Agents
                  and MultiAgent Systems, {AAMAS} 2018, Stockholm, Sweden, July 10-15,
                  2018},
  pages        = {1362--1370},
  publisher    = {International Foundation for Autonomous Agents and Multiagent Systems
                  Richland, SC, {USA} / {ACM}},
  year         = {2018},
  url          = {http://dl.acm.org/citation.cfm?id=3237903},
  bibsource    = {dblp computer science bibliography, https://dblp.org}
}

@inproceedings{bullinger24robustpopular,
  author       = {Martin Bullinger and
                  Rohith Reddy Gangam and
                  Parnian Shahkar},
  editor       = {Mehdi Dastani and
                  Jaime Sim{\~{a}}o Sichman and
                  Natasha Alechina and
                  Virginia Dignum},
  title        = {Robust Popular Matchings},
  booktitle    = {Proceedings of the 23rd International Conference on Autonomous Agents
                  and Multiagent Systems, {AAMAS} 2024, Auckland, New Zealand, May 6-10,
                  2024},
  pages        = {225--233},
  publisher    = {International Foundation for Autonomous Agents and Multiagent Systems
                  / {ACM}},
  year         = {2024},
  url          = {https://dl.acm.org/doi/10.5555/3635637.3662870},
  doi          = {10.5555/3635637.3662870},
  bibsource    = {dblp computer science bibliography, https://dblp.org}
}

@article{zhou2019robustonline,
 author       = {Yu{-}Hang Zhou and
                  Chen Liang and
                  Nan Li and
                  Cheng Yang and
                  Shenghuo Zhu and
                  Rong Jin},
  title        = {Robust Online Matching with User Arrival Distribution Drift},
  booktitle    = {The Thirty-Third {AAAI} Conference on Artificial Intelligence, {AAAI}
                  2019, The Thirty-First Innovative Applications of Artificial Intelligence
                  Conference, {IAAI} 2019, The Ninth {AAAI} Symposium on Educational
                  Advances in Artificial Intelligence, {EAAI} 2019, Honolulu, Hawaii,
                  USA, January 27 - February 1, 2019},
  pages        = {459--466},
  publisher    = {{AAAI} Press},
  year         = {2019},
  url          = {https://doi.org/10.1609/aaai.v33i01.3301459},
  doi          = {10.1609/AAAI.V33I01.3301459},
  bibsource    = {dblp computer science bibliography, https://dblp.org}
}

@inproceedings{bian2022robustsubset,
   author       = {Chao Bian and
                  Yawen Zhou and
                  Chao Qian},
  editor       = {Luc De Raedt},
  title        = {Robust Subset Selection by Greedy and Evolutionary Pareto Optimization},
  booktitle    = {Proceedings of the Thirty-First International Joint Conference on
                  Artificial Intelligence, {IJCAI} 2022, Vienna, Austria, 23-29 July
                  2022},
  pages        = {4726--4732},
  publisher    = {ijcai.org},
  year         = {2022},
  url          = {https://doi.org/10.24963/ijcai.2022/655},
  doi          = {10.24963/IJCAI.2022/655},
  bibsource    = {dblp computer science bibliography, https://dblp.org}
}

@inproceedings{rangapuram15team,
 author       = {Syama Sundar Rangapuram and
                  Thomas B{\"{u}}hler and
                  Matthias Hein},
  editor       = {Daniel Schwabe and
                  Virg{\'{\i}}lio A. F. Almeida and
                  Hartmut Glaser and
                  Ricardo Baeza{-}Yates and
                  Sue B. Moon},
  title        = {Towards realistic team formation in social networks based on densest
                  subgraphs},
  booktitle    = {22nd International World Wide Web Conference, {WWW} '13, Rio de Janeiro,
                  Brazil, May 13-17, 2013},
  pages        = {1077--1088},
  publisher    = {International World Wide Web Conferences Steering Committee / {ACM}},
  year         = {2013},
  url          = {https://doi.org/10.1145/2488388.2488482},
  doi          = {10.1145/2488388.2488482},
  bibsource    = {dblp computer science bibliography, https://dblp.org}
}

@inproceedings{StelmakhSS19,
  author       = {Ivan Stelmakh and
                  Nihar B. Shah and
                  Aarti Singh},
  editor       = {Aur{\'{e}}lien Garivier and
                  Satyen Kale},
  title        = {PeerReview4All: Fair and Accurate Reviewer Assignment in Peer Review},
  booktitle    = {Algorithmic Learning Theory, {ALT} 2019, 22-24 March 2019, Chicago,
                  Illinois, {USA}},
  series       = {Proceedings of Machine Learning Research},
  volume       = {98},
  pages        = {827--855},
  publisher    = {{PMLR}},
  year         = {2019},
  url          = {http://proceedings.mlr.press/v98/stelmakh19a.html},
  bibsource    = {dblp computer science bibliography, https://dblp.org}
}

\end{document}